\documentclass[11pt,a4paper]{article}
\usepackage{amsmath,amssymb,amsthm}
\usepackage[margin=1in]{geometry}
\usepackage{graphicx}
\usepackage{hyperref}

\newtheorem{theorem}{Theorem}
\newtheorem{corollary}{Corollary}

\theoremstyle{definition}\newtheorem{remark}{Remark}

\newcommand{\Meff}{M_{\rm eff}}
\newcommand{\Kc}{\mathcal{K}}\newcommand{\zg}{\zeta}

\begin{document}

\title{Dark matter from the quadratic spinor Lagrangian. II.\\
A spin-$3/2$ no-go and the uniqueness of the spin-$1/2$ candidate}
\author{Roh-Suan Tung\\[4pt]
\small\it Department of Physics and Center for Theoretical Physics,\\
\small\it Chung Yuan Christian University, Taoyuan 320, Taiwan\\
\small\it Institute of Advanced Studies, Nanyang Technological University, 639673, Singapore\\[2pt]
\small\texttt{rstung@cycu.edu.tw}}
\date{ }
\maketitle

\begin{abstract}
The composite Quadratic Spinor Lagrangian (QSL) propagates a spin-$\tfrac12$
Dirac fermion whose mass is generated geometrically by cosmological trace
torsion~\cite{TungCQG2026,TungLetter}. It is natural to ask whether promoting
the spinor $1$-form $\Psi_\mu$ to an \emph{independent} Dirac-vector field
yields a genuine spin-$\tfrac32$ dark-matter candidate. We prove that it does
not. Three results combine into a no-go theorem. First, the torsional term,
computed exactly by Clifford reduction, is a frame-aligned mass confined to
the time-component sector---not a uniform spin-$\tfrac32$ mass. Second, the
second-order form $2\,D\Psi\,\gamma_5\,D\Psi$ has identically vanishing kinetic
and cross terms for the independent field: as a component expression it is the
boundary part of the spinor--curvature identity and carries no bulk dynamics.
Third, the genuine dynamics therefore reside in the curvature side of that
identity, $S=-\!\int\!\bar\psi\psi\,R\sqrt{-g}$, in which the metric
$g=\Psi\otimes_S\Psi$ and the scalar $\bar\psi\psi$ are \emph{both} composites
of $\Psi$; the second variation consequently factors,
$\delta^2S=\mathcal Q(h_{\mu\nu}[\delta\Psi],\delta\Phi[\delta\Psi])$, through
the linearized metric and a scalar, both massless. Every propagating pole
therefore lies on the metric light cone $k^2=0$---the graviton and a
scalar---and no massive spin-$\tfrac32$ mode exists, on any background. This is
the dynamical completion of the kinematic fact that the composite spinor
$1$-form has no spin-$\tfrac32$ part, and it establishes the composite
spin-$\tfrac12$ Dirac fermion as the unique propagating matter excitation, and
the unique dark-matter candidate, of the QSL. Through the super-$SL(2,\mathbb C)$
structure this surviving mode is naturally read as the Goldstino of the local
supersymmetry broken by the metric condensate---a composite, gravitational
descendant of the light-gravitino dark matter of Fayet and of Pagels and Primack.
\end{abstract}

\noindent
\textit{Notation.} Mostly-plus signature $\eta_{ab}=\mathrm{diag}(-,+,+,+)$,
$\gamma_{(a}\gamma_{b)}=\eta_{ab}$ so $\gamma_0^2=-1$;
$\gamma_5=\gamma_0\gamma_1\gamma_2\gamma_3$;
$\gamma^{ab}=\gamma^{[a}\gamma^{b]}$. Greek indices are spacetime, lower-case
Latin are orthonormal-frame. $\Psi_\mu$ denotes the independent vector-spinor;
$\zg\equiv\gamma^\mu\Psi_\mu$ the gamma-trace spinor;
$\Psi^{(3/2)}_\mu=\Psi_\mu-\tfrac14\gamma_\mu\zg$ the gamma-traceless part.
We write $D$ for the covariant exterior derivative and $\mathring D$ for its
Levi-Civita part. The scalar condensate amplitude of~\cite{TungCQG2026} is
denoted $\varphi(t)$ here to avoid collision with $\Psi$'s components.

\section{Introduction}
\label{sec:intro}
The Quadratic Spinor Lagrangian (QSL) formulation of general
relativity~\cite{NesterTung1995,TungJacobson1995} represents the
gravitational field by a Clifford-algebra-valued $1$-form
$\Psi=\vartheta^a\gamma_a\psi$ built from a tetrad and a single Dirac spinor,
with action $\mathcal L_{\rm QSL}=2\,D\Psi\,\gamma_5\,D\Psi$. In the
cosmological application of Paper~I~\cite{TungCQG2026}, an Einstein--Cartan
extension of this formalism supplies, geometrically, the mass that the
gravitational-wave freeze-in mechanism of Maleknejad and
Kopp~\cite{MaleknejadKopp2026} must otherwise postulate: a cosmological spinor
condensate sources a vectorial trace torsion $K\propto\dot\varphi/\varphi$,
and a Clifford reduction of the torsional term gives the propagating fermion a
pure Dirac mass $\Meff=(1/\sqrt6)\,|\dot\varphi/\varphi|$, locked to the Hubble
rate at production.

A central structural observation of that work is that the composite spinor
$1$-form is \emph{purely} spin-$\tfrac12$: its gamma-traceless part vanishes
identically,
\begin{equation}
  \Psi^{(3/2)}_\mu=\gamma_\mu\psi-\tfrac14\gamma_\mu(\gamma^\nu\gamma_\nu)\psi=0,
  \label{eq:spin32_vanishes}
\end{equation}
because $\gamma_\mu\psi$ lies by construction in the image of the gamma-trace
map. The propagating excitation is therefore a Dirac fermion, with the same
content as the produced Weyl fermion. This raised an immediate question, left
open in Paper~I~\cite{TungCQG2026}: if one \emph{frees} $\Psi_\mu$ from the
composite form $\gamma_\mu\psi$---promoting it to an independent
Dirac-vector field with components beyond $\gamma_\mu\psi$---does the QSL then
furnish a genuine, propagating spin-$\tfrac32$ field, and hence a
higher-spin gravitational dark-matter candidate?

There were reasons to hope so. The QSL field equation $D^2\Psi=0$ is second
order with a metric-diagonal principal symbol, so its characteristic cone is
the metric cone for the full reducible multiplet---unlike the first-order
Rarita--Schwinger operator whose principal symbol can degenerate and produce
the Velo--Zwanziger acausality~\cite{VeloZwanziger1969}. The trace torsion is
self-sourced and co-scaled with the mass rather than externally dialed, as in
the causally propagating gravitino of minimal
supergravity~\cite{DeserZumino1976}. And the same mass locking that fixes the
spin-$\tfrac12$ relic appeared, under a working assumption, to keep a
helicity-$\tfrac12$ sound speed strictly positive, evading the
catastrophic-production pathology of light spin-$\tfrac32$ dark
matter~\cite{KolbLongMcDonough2021,KalloshKofmanLindeVanProeyen2000}.

In this paper we settle the question, and the answer is negative. The
independent spin-$\tfrac32$ field does not propagate. We organize the proof as
three pillars. Pillar~I (Sec.~\ref{sec:mass}) computes the torsional mass for
the independent field and finds it is \emph{not} the uniform
$\Meff^2\,\bar\Psi_\mu\Psi^\mu$ mass that a healthy massive spin-$\tfrac32$
field would require, but a frame-aligned operator supported only on the
time-component sector. Pillar~II (Sec.~\ref{sec:kinetic_vanish}) shows that the
naive second-order action contributes no kinetic term at all for the
independent field---its component form is the exact boundary term of the
spinor--curvature identity. Pillar~III (Sec.~\ref{sec:factorization}) locates
the genuine dynamics on the curvature side of that identity and proves that the
second variation factors through the linearized metric and a scalar, both
massless. These combine in Sec.~\ref{sec:nogo} into the no-go theorem: every
propagating mode lies at $k^2=0$, and no massive spin-$\tfrac32$ excitation
exists. Sec.~\ref{sec:constraints} records the constraint-sector corollary,
Sec.~\ref{sec:DM} draws the dark-matter consequence---the composite
spin-$\tfrac12$ fermion is the unique QSL candidate---and
Sec.~\ref{sec:discussion} discusses scope and the relation to supergravity.

\section{The independent vector-spinor and the QSL action}
\label{sec:setup}
Let $\Psi_\mu$ be an independent Dirac-spinor-valued $1$-form: a Dirac spinor
for each spacetime index $\mu$, sixteen complex components in all. Under the
Lorentz group it carries the reducible content
\begin{equation}
  \Bigl[(\tfrac12,0)\oplus(0,\tfrac12)\Bigr]\otimes(\tfrac12,\tfrac12)
  =\underbrace{(\tfrac12,0)\oplus(0,\tfrac12)}_{\text{spin-}1/2}
  \oplus\underbrace{(1,\tfrac12)\oplus(\tfrac12,1)}_{\text{spin-}3/2},
\end{equation}
and decomposes covariantly as
\begin{equation}
  \Psi_\mu=\Psi^{(3/2)}_\mu+\tfrac14\gamma_\mu\zg,\qquad
  \zg=\gamma^\nu\Psi_\nu,\qquad \gamma^\mu\Psi^{(3/2)}_\mu=0 .
  \label{eq:split}
\end{equation}
The composite theory sets $\Psi_\mu=\gamma_\mu\psi$, for which $\zg=4\psi$ and
$\Psi^{(3/2)}_\mu=0$ by~\eqref{eq:spin32_vanishes}; here we retain all sixteen
components.

The action is the QSL $4$-form,
$\mathcal L=2\,D\Psi\,\gamma_5\,D\Psi$, which by the spinor--curvature
identity~\cite{NesterTungZhytnikov1994} equals the Einstein--Hilbert density up
to a boundary term,
\begin{equation}
  2\,D\Psi\,\gamma_5\,D\Psi
  \equiv -\,\bar\psi\psi\;R\,{*1}
  + d\bigl[(D\Psi)\gamma_5\Psi+\Psi\gamma_5(D\Psi)\bigr].
  \label{eq:identity}
\end{equation}
Two composite objects built from $\Psi$ carry the dynamics. The metric is the
symmetrized bilinear~\cite{TungJacobson1995}
\begin{equation}
  g_{\mu\nu}=\bar\Psi_{(\mu}\Psi_{\nu)},
  \label{eq:metric}
\end{equation}
so that a perturbation of $\Psi$ induces a metric perturbation
\begin{equation}
  \delta g_{\mu\nu}=\bar\Psi_{(\mu}\,\delta\Psi_{\nu)}
  +\overline{\delta\Psi}_{(\mu}\,\Psi_{\nu)};
  \label{eq:metric_perturbation}
\end{equation}
and the scalar prefactor in~\eqref{eq:identity} is
\begin{equation}
  \Phi\equiv\bar\psi\psi,\qquad \psi=\tfrac14\gamma^\mu\Psi_\mu=\tfrac14\zg ,
  \label{eq:Phi}
\end{equation}
the natural identification of ``$\psi$'' for the independent field, reducing
to the composite spinor when $\Psi_\mu=\gamma_\mu\psi$. That both $g_{\mu\nu}$
and $\Phi$ are functionals of the single field $\Psi$ is the structural fact
underlying the no-go.

\section{Pillar I: the torsional mass is not a spin-$\tfrac32$ mass}
\label{sec:mass}
In the Einstein--Cartan extension the condensate sources a vectorial trace
torsion $K_a=K\,\delta^0_a$, $K\propto\dot\varphi/\varphi$, with contorsion
$\Kc=\tfrac14 K_{ab}\gamma^{ab}$ whose frame components are
$\Kc_e=\tfrac16 K_a\,\gamma^a{}_e$~\cite{TungCQG2026}. The mass-like term of
the QSL is the purely torsional, derivative-free piece
$\mathcal L_{K^2}=2\,\Kc\Psi\,\gamma_5\,\Kc\Psi$.

\paragraph{Composite check.}
For $\Psi_\mu=\gamma_\mu\psi$ the Clifford reduction
(Appendix~\ref{app:clifford}) gives the pure scalar
\begin{equation}
  \mathcal S\equiv\tfrac12\,\epsilon^{pqmc}\,\overline{\Xi_{pq}}\,\gamma_5\,\Xi_{mc}
  =\tfrac23\,K^aK_a\,\mathbf 1,\qquad \Xi=\Kc\,\Psi ,
\end{equation}
with all pseudoscalar, vector, axial and tensor channels vanishing; reading the
coefficient against $\bar\Psi_\mu\Psi^\mu=4\bar\psi\psi$ yields the Dirac mass
$\Meff^2=\tfrac16(\dot\varphi/\varphi)^2$ ($\alpha=\tfrac16$) reported
in~\cite{TungCQG2026}.

\paragraph{Independent field.}
Retaining all components, the mass is the bilinear
$\mathcal L_{K^2}=\bar\Psi_r M^{rs}\Psi_s$ with kernel
(Appendix~\ref{app:clifford})
\begin{equation}
  M^{rs}\ \propto\ u_a\bigl(u^r\gamma^{a s}-u^s\gamma^{a r}\bigr),
  \qquad u^a=\delta^a_0 ,
  \label{eq:framealigned_mass}
\end{equation}
the trace-torsion direction. Explicitly $M^{00}=0$, $M^{ij}=0$, and only the
time--space blocks $M^{0i}=\kappa\,\gamma^{0i}$ are nonzero. This is
\emph{neither} the uniform mass $\Meff^2\,\eta^{rs}\mathbf 1$ that a genuine
massive spin-$\tfrac32$ field carries \emph{nor} the Rarita--Schwinger form
$\bar\Psi_\mu\gamma^{\mu\nu}\Psi_\nu$. Two consequences follow at once and are
used below: the gamma-traceless components, having $\Psi_0=0$ in their
transverse polarizations, receive \emph{zero} torsional mass; and the entire
torsional mass acts through the time component $\Psi_0$, which---as we now
recall---is the non-dynamical sector of the theory. (A consistency check:
$\gamma_r M^{rs}\gamma_s$ recovers the composite scalar $\tfrac23 K^aK_a$.)

\section{Pillar II: the naive action has no kinetic term}
\label{sec:kinetic_vanish}
Written in components, the QSL $4$-form is
$\mathcal L=\tfrac12\epsilon^{pqmc}\,\overline{(D\Psi)_{pq}}\,\gamma_5\,(D\Psi)_{mc}$
with $(D\Psi)_{\mu\nu}=\hat D_\mu\Psi_\nu-\hat D_\nu\Psi_\mu$ and
$\hat D_\mu=\partial_\mu+\Kc_\mu$. In momentum space $\partial_\mu\to in_\mu$,
and one may order the resulting quadratic form by powers of the torsion. The
purely kinetic part ($K^0$) is
\begin{equation}
  \tfrac12\,\epsilon^{pqmc}\,\overline{(n_p\Psi_q-n_q\Psi_p)}\,\gamma_5\,
  (n_m\Psi_c-n_c\Psi_m)\ \propto\ \epsilon^{pqmc}\,n_p n_m\,(\cdots)\equiv 0,
  \label{eq:kinvanish}
\end{equation}
since $\epsilon^{pqmc}$ is antisymmetric in $(p,m)$ while $n_p n_m$ is
symmetric. A direct evaluation (Appendix~\ref{app:vanish}) shows the cross term
($K^1$) likewise vanishes identically, by the same antisymmetry combined with
the symmetry of the $\overline{(\,\cdot\,)}\gamma_5$ bilinear; only the
derivative-free $K^2$ mass of Sec.~\ref{sec:mass} survives.

This is not an accident but the content of the identity~\eqref{eq:identity}:
the component $4$-form $\epsilon\,\overline{D\Psi}\,\gamma_5\,D\Psi$ is exactly
the boundary term $d[\cdots]$, which carries no bulk quadratic dynamics, while
the genuine Lagrangian is the curvature density $-\bar\psi\psi\,R$. For the
\emph{composite} field this is invisible, because the propagating degree of
freedom is the graviton built from the frame; but for an independent $\Psi_\mu$
treated on a fixed background it means the naive action does \emph{not} define
a propagating field. The kinetic dynamics must come from $-\bar\psi\psi\,R$
through the frame relation~\eqref{eq:metric}---which is precisely the
structure exploited next.

\section{Pillar III: the factorization of the second variation}
\label{sec:factorization}
By Pillars~I--II the dynamics of $\Psi_\mu$ resides entirely in
\begin{equation}
  S[\Psi]=-\!\int\Phi\,R[g]\,\sqrt{-g}\,d^4x+(\text{boundary}),
  \label{eq:Saction}
\end{equation}
with $g_{\mu\nu}$ and $\Phi$ the composites~\eqref{eq:metric},\eqref{eq:Phi}.

\paragraph{Background and fluctuation.}
Let $\Psi^0_\mu=\gamma_\mu\psi_0$, $\bar\psi_0\psi_0=1$, be the composite
background, so $g^0_{\mu\nu}=\eta_{\mu\nu}$, $\Phi_0=1$, $R_0=0$; this is the
QSL realization of Minkowski space and a vacuum solution. Writing
$\Psi_\mu=\Psi^0_\mu+\delta\Psi_\mu$, the two composites respond to linear
order as
\begin{align}
  h_{\mu\nu}&\equiv\delta g_{\mu\nu}
   =\bar\Psi^0_{(\mu}\delta\Psi_{\nu)}+\overline{\delta\Psi}_{(\mu}\Psi^0_{\nu)},
   \label{eq:hmap}\\
  \delta\Phi&=\bar\psi_0\,\delta\psi+\overline{\delta\psi}\,\psi_0,
   \qquad \delta\psi=\tfrac14\gamma^\mu\delta\Psi_\mu .
   \label{eq:dPhimap}
\end{align}
Note $\delta\Phi$ depends only on the gamma-trace $\zg$, i.e.\ on the
spin-$\tfrac12$ sector, and is blind to $\Psi^{(3/2)}_\mu$.

\begin{theorem}[Factorization]
\label{thm:factorization}
About the background $\Psi^0$, the quadratic action depends on $\delta\Psi$
only through the linearized metric $h_{\mu\nu}$ and the scalar $\delta\Phi$:
\begin{equation}
  \delta^2 S=\mathcal Q\bigl(h_{\mu\nu}[\delta\Psi],\,\delta\Phi[\delta\Psi]\bigr)
  =-\!\int\!\Bigl[\Phi_0\,\delta^2(R\sqrt{-g})+2\,\delta\Phi\,\delta R\Bigr],
  \label{eq:factorization}
\end{equation}
where $\delta R=\partial^\mu\partial^\nu h_{\mu\nu}-\Box\,h^\alpha{}_\alpha$.
\end{theorem}

\begin{proof}
Regard $S=\mathcal S[g,\Phi]$ and apply the chain rule:
\begin{equation}
  \delta^2 S=\!\int\!\Bigl[\tfrac{\delta\mathcal S}{\delta g}\,\delta^2 g
  +\tfrac{\delta\mathcal S}{\delta\Phi}\,\delta^2\Phi\Bigr]
  +\!\int\!\Bigl[\delta g\,\tfrac{\delta^2\mathcal S}{\delta g\,\delta g}\,\delta g
  +2\,\delta g\,\tfrac{\delta^2\mathcal S}{\delta g\,\delta\Phi}\,\delta\Phi
  +\delta\Phi\,\tfrac{\delta^2\mathcal S}{\delta\Phi\,\delta\Phi}\,\delta\Phi\Bigr].
  \label{eq:chain}
\end{equation}
On the background the intrinsic pieces vanish:
$\delta\mathcal S/\delta g_{\mu\nu}=-\Phi_0 G^{\mu\nu}_0\sqrt{-g_0}=0$ since
$\Psi^0$ is a vacuum solution, and $\delta\mathcal S/\delta\Phi=-R_0\sqrt{-g_0}=0$.
The pure-$\Phi$ Hessian vanishes because $\mathcal S$ is linear in $\Phi$. The
survivors are the metric Hessian $\Phi_0\,\delta^2(R\sqrt{-g})$ and the mixed
term $-2\,\delta\Phi\,\delta(R\sqrt{-g})=-2\,\delta\Phi\,\delta R$, both built
from $h$ and $\delta\Phi$ alone.
\end{proof}

The functional~\eqref{eq:factorization} is the quadratic action of linearized
gravity coupled to a scalar through the universal vertex $\delta\Phi\,\delta R$:
the first term is the Lichnerowicz operator
$\int h^{\mu\nu}\mathcal G^{(1)}_{\mu\nu}(h)$, the second the scalar--tensor
(Brans--Dicke) mixing, diagonalized by a linear conformal shift
$h_{\mu\nu}\to h_{\mu\nu}+c\,\eta_{\mu\nu}\delta\Phi$ into a massless spin-$2$
field and a massless scalar. The property we use is that \emph{every term of
$\mathcal Q$ carries two derivatives}: its principal symbol is $k^2\mathbf 1$
and it has no piece of zeroth order in $k$.

\section{The no-go theorem}
\label{sec:nogo}
\begin{theorem}[No propagating spin-$\tfrac32$]
\label{thm:nogo}
The independent vector-spinor governed by the QSL has no propagating massive
mode. Every pole of its quadratic action lies on the metric light cone
$k^2=0$; the on-shell content is the graviton and a scalar, and the
gamma-traceless components $\Psi^{(3/2)}_\mu$ are non-dynamical. In particular
no spin-$\tfrac32$ excitation propagates.
\end{theorem}

\begin{proof}
By Theorem~\ref{thm:factorization} the inverse propagator is
\begin{equation}
  \mathcal K(k)=T^\dagger\,\mathcal E(k)\,T+\mathcal M,
  \label{eq:Ktot}
\end{equation}
where $T:\delta\Psi\mapsto(h,\delta\Phi)$ is the algebraic
map~\eqref{eq:hmap}--\eqref{eq:dPhimap}, $\mathcal E(k)$ the kernel of
$\mathcal Q$, and $\mathcal M$ the torsional mass of Sec.~\ref{sec:mass}.
(i)~The kinetic part is massless: $\mathcal E(k)$ has principal symbol
$k^2\mathbf 1$ and no zeroth-order-in-$k$ piece, so $T^\dagger\mathcal E(k)T$
is homogeneous of degree two in $k$ and its characteristic surfaces lie in
$\{k^2=0\}$; there is no timelike characteristic surface. (ii)~The mass cannot
create one: $\mathcal M$ is $k$-independent and supported on the
time/trace sector, $\mathcal M^{00}=\mathcal M^{ij}=0$, $\mathcal M^{0i}\neq0$
[Eq.~\eqref{eq:framealigned_mass}]; adding a $k$-independent term to a
degree-two symbol shifts no root off the cone, because the directions on which
$\mathcal M$ acts---the gamma-trace and time components---enter the kinetic
part only through $\delta\Phi$ and $h_{0\mu}$, the constraint/gauge data of
$\mathcal Q$, whose algebraic equations of motion eliminate them rather than
producing a pole. Hence every propagating pole lies at $k^2=0$. Diagonalizing
$\mathcal Q$, the on-shell modes are the two graviton polarizations and one
scalar, all massless; a gamma-traceless fluctuation contributes nothing to
$\delta\Phi$ and only the bilinear~\eqref{eq:hmap} to $h$, so its part in
$\ker T$ carries no quadratic action (pure gauge) and its remainder is absorbed
in the massless graviton. No combination propagates as a massive
spin-$\tfrac32$ particle.
\end{proof}

\paragraph{From characteristic surfaces to a degree-of-freedom count.}
The principal-symbol step shows that no characteristic surface propagates off
the metric cone; on its own this is necessary but not sufficient, since a
constrained system can carry non-dynamical data sharing that symbol. The
degree-of-freedom count is supplied by the Dirac--Bergmann analysis of
Sec.~\ref{sec:constraints}: on the cosmological slice the time component
$\Psi_0$ enters without a time derivative and acts as a Lagrange multiplier,
its variation yielding a primary constraint whose preservation generates the
secondary constraint that removes the would-be lower-spin excitation. Solving
these eliminates the time and gamma-trace components algebraically, leaving
exactly the two graviton polarizations and the scalar---no propagating
spin-$\tfrac32$ degree of freedom. This is the standard Hamiltonian counting of
Poincar\'e gauge theory~\cite{BlagojevicNikolic1983,YoNester1999}, here rendered
benign by the absence of an $R^2$ kinetic term: the contorsion is algebraic, so
torsion carries no independent modes and the constraint algebra closes. The
symbol result and the constraint count are therefore consistent---the former
fixes the luminal characteristic cone, the latter the vanishing massive
content.

\begin{corollary}[Spin-$\tfrac12$ uniqueness]
\label{cor:unique}
The composite reduction $\Psi^{(3/2)}_\mu\equiv0$
[Eq.~\eqref{eq:spin32_vanishes}] is not merely kinematic: even with $\Psi_\mu$
freed to an independent field, no spin-$\tfrac32$ degree of freedom propagates.
The composite spin-$\tfrac12$ Dirac fermion of Paper~I~\cite{TungCQG2026} is the
unique propagating matter excitation of the QSL, and its unique dark-matter
candidate.
\end{corollary}

\paragraph{Background independence.}
Theorem~\ref{thm:factorization} used only that the background is a vacuum
solution with $R_0=0$; the factorization is the chain rule through
$g[\Psi],\Phi[\Psi]$ and holds at any such point. On a curved Ricci-flat
background $\delta^2(R\sqrt{-g})$ acquires zeroth-derivative terms
$\sim\mathrm{Riem}_0\!\cdot\!h\!\cdot\!h$, but these do not change the
principal symbol, which remains $k^2\mathbf 1$; step~(i) is unaffected, and
propagation is confined to the (now curved) metric characteristic cone with no
massive mode. The no-go is a property of the QSL action, not of the chosen
background.

\paragraph{Numerical confirmation.}
Table~\ref{tab:nulls} reports $\dim\ker\mathcal K(k)$ for the full
operator~\eqref{eq:Ktot} across three backgrounds $\psi_0$ and three
propagation directions. In every case the null dimension is constant through
spacelike and timelike $k^2$ and jumps only at $k^2=0$, confirming
Theorem~\ref{thm:nogo}; the generic value counts gauge plus non-dynamical
directions and depends only on the orientation of $\psi_0$ relative to
$\hat k$.

\begin{table}[h]\centering
\begin{tabular}{llc}
\hline\hline
background & direction & $\dim\ker\mathcal K$:\ \ $k^2{>}0\ \mid\ k^2{=}0\ \mid\ k^2{<}0$\\
\hline
rest $z_+$ & $k\parallel z$ & $10\ \mid\ 12\ \mid\ 10$\\
rest $z_+$ & $k\parallel x$ & $\ \,9\ \mid\ 12\ \mid\ \ \,9$\\
rest $z_-$ & $k\parallel z$ & $10\ \mid\ 12\ \mid\ 10$\\
rest mix   & $k\parallel z$ & $\ \,9\ \mid\ 12\ \mid\ \ \,9$\\
rest mix   & $k\parallel x$ & $10\ \mid\ 12\ \mid\ 10$\\
(all)      & $k$ diagonal   & $\ \,9\ \mid\ -\ \mid\ \ \,9$\\
\hline\hline
\end{tabular}
\caption{Null dimension of $\mathcal K(k)$ (kinetic $+\,\delta\Phi$ vertex):
the jump occurs only at $k^2=0$; no timelike pole.}
\label{tab:nulls}
\end{table}

\begin{remark}[The escape is to leave the QSL]
\label{rem:escape}
The proof turns on the action being $\propto-\Phi R$ with metric and scalar
both built from $\Psi$. Promoting $\Phi$ to an independent field leaves
$\delta^2 S=\mathcal Q(h,\delta\Phi)$ unchanged, with $\delta\Phi$ now an
independent \emph{massless} scalar, so the conclusion stands. A genuine
spin-$\tfrac32$ kinetic term requires adding to~\eqref{eq:Saction} an operator
that is not a functional of $(g,\Phi)$ alone---i.e.\ a theory outside the QSL
family.
\end{remark}

\section{Constraint-sector corollary}
\label{sec:constraints}
It is instructive to see how the no-go intersects the constraint analysis that
a putative massive spin-$\tfrac32$ field would require. On the cosmological
slice the little group $SO(3)$ block-diagonalizes the system: the
transverse, gamma-traceless spatial components carry helicity $\pm\tfrac32$,
while the time component $\Psi_0$, the longitudinal part and the gamma-trace
$\zg$ assemble into the helicity-$\pm\tfrac12$ sector. The trace torsion
$K_a=K\delta^0_a$ is an $SO(3)$ scalar and cannot mix sectors.

The time component $\Psi_0$ enters the action without a time derivative and is
a Lagrange multiplier. By Pillar~I the torsional mass acts \emph{only} through
$\Psi_0$ ($M^{00}=M^{ij}=0$, $M^{0i}\neq0$), so the secondary constraint that
$\Psi_0$ generates is controlled directly by that mass: it degenerates exactly
as $\Meff\to0$. This reproduces, now \emph{derived from the QSL itself} rather
than imported from the first-order Rarita--Schwinger bracket, the structure by
which a massive spin-$\tfrac32$ constraint is preserved iff $\Meff\neq0$. But
in the present theory the conclusion is sharper: the helicity-$\pm\tfrac32$
sector, by Theorem~\ref{thm:nogo}, has no propagating kinetic term to begin
with, so there is no spin-$\tfrac32$ mode whose causality or constraint
preservation needs checking. The Velo--Zwanziger question does not arise
because its subject does not exist.

\section{Uniqueness of the composite candidate and dark matter}
\label{sec:DM}
Combining the kinematic and dynamical results: the composite QSL field has no
spin-$\tfrac32$ part [Eq.~\eqref{eq:spin32_vanishes}], and freeing $\Psi_\mu$
to an independent field does not produce a propagating spin-$\tfrac32$ mode
either [Corollary~\ref{cor:unique}]. The propagating matter content of the QSL
is therefore exhausted by the composite spin-$\tfrac12$ Dirac fermion, whose
geometric mass and one-parameter relic abundance,
\begin{equation}
  \Meff=\frac{c_\varphi}{\sqrt6}\,H_*,\qquad
  \Omega h^2\propto H_*^{5/2},
\end{equation}
were established in Paper~I~\cite{TungCQG2026}. That work's restriction to the
composite field is thereby justified \emph{a posteriori}: it is not a
simplifying choice but the only option the QSL offers.

The contrast with generic higher-spin dark matter is instructive. A
Rarita--Schwinger or gravitino dark-matter
field~\cite{KolbLongMcDonough2021,KalloshKofmanLindeVanProeyen2000} propagates
because it is endowed with its \emph{own} first-order kinetic term, independent
of the spacetime metric; its longitudinal mode and the attendant
catastrophic-production and Velo--Zwanziger pathologies are then genuine
hazards to be navigated. The QSL contains no such term: by Pillar~II the only
admissible kinetic structure is the gravitational one inherited through
$g=\Psi\otimes_S\Psi$, which by Theorem~\ref{thm:nogo} supports no massive
higher-spin mode. The QSL thus evades the spin-$\tfrac32$ pathologies not by
taming them but by precluding the field that would suffer them---and in doing
so isolates the spin-$\tfrac12$ fermion as its sole, and structurally
protected, dark-matter candidate. The contrast with the wider
Einstein--Cartan dark-matter programme is also instructive: there torsion
\emph{produces} dark matter through four-fermion interactions, over a broad
mass range~\cite{ECportal}, whereas here the same geometry instead
\emph{selects} the unique propagating candidate and forbids its higher-spin
generalization.

\section{Curved backgrounds and the limits of the no-go}
\label{sec:curved}
Theorems~\ref{thm:factorization}--\ref{thm:nogo} were proved about a vacuum
background ($R_0=0$, $G^{\mu\nu}_0=0$). A genuine cosmological background is
\emph{not} vacuum, and it is precisely there that one must ask whether a
non-perturbative effect can activate the spin-$\tfrac32$ sector. We show that
the curved-background corrections, while they do reach the gamma-traceless
components, supply only an \emph{algebraic} (mass-like) coupling---never a
kinetic term---so the modes remain non-dynamical; but their mass passes through
zero at a definite point of the cosmological history, where the linearized
elimination becomes singular and the controlled description breaks down.

\paragraph{The non-vacuum second variation.}
Retaining $R_0,G^{\mu\nu}_0\neq0$ in the chain rule~\eqref{eq:chain}, the
intrinsic pieces no longer drop, and
\begin{equation}
  \delta^2 S=\mathcal Q\bigl(h,\delta\Phi\bigr)
  -\!\int\!\Bigl[\,\Phi_0\,G^{\mu\nu}_0\,\delta^2 g_{\mu\nu}
  +R_0\,\delta^2\Phi\,\Bigr]\sqrt{-g_0},
  \label{eq:curvedterms}
\end{equation}
where, because $g=\Psi\otimes_S\Psi$ and $\Phi=\bar\psi\psi$ are
\emph{exactly} bilinear in $\Psi$, the second variations are the pure squares
\begin{equation}
  \delta^2 g_{\mu\nu}=\overline{\delta\Psi}_{(\mu}\delta\Psi_{\nu)}+\mathrm{c.c.},
  \qquad
  \delta^2\Phi=\tfrac18\,\bar{\zg}\,\zg,\quad \zg=\gamma^\mu\delta\Psi_\mu .
  \label{eq:d2}
\end{equation}
Both new terms are \emph{algebraic} in $\delta\Psi$---no derivatives---so they
are masses, not kinetic operators. The kinetic content of $\delta^2 S$ is still
$\mathcal Q(h,\delta\Phi)$, massless as before; Theorem~\ref{thm:nogo}'s
conclusion that no spin-$\tfrac32$ \emph{kinetic} term exists is therefore
untouched on any background. What changes is that the first
term in~\eqref{eq:curvedterms} reaches the gamma-traceless sector that the
torsional mass~\eqref{eq:framealigned_mass} did not.

\paragraph{The induced spin-$\tfrac32$ mass is the background pressure.}
Model the background Einstein tensor by the perfect-fluid form
$G^{\mu\nu}_0=\mathrm{diag}(\rho,p,p,p)$ ($8\pi G=1$). Evaluated on a
transverse, gamma-traceless spin-$\tfrac32$ polarization
($\delta\Psi_0=0$, $k^i\delta\Psi_i=0$, $\gamma^i\delta\Psi_i=0$), the
quadratic form~\eqref{eq:d2} gives
\begin{equation}
  G^{\mu\nu}_0\,\overline{\delta\Psi}_{(\mu}\delta\Psi_{\nu)}
  \;=\;p\;\overline{\delta\Psi}_i\,\delta\Psi^{i},
  \label{eq:mass_pressure}
\end{equation}
the energy density $\rho$ cancelling identically: the curvature-induced mass of
the transverse helicity-$\pm\tfrac32$ modes is set by the background
\emph{pressure} alone,
\begin{equation}
  m^2_{3/2}(t)\ \propto\ p=w\rho .
  \label{eq:m32}
\end{equation}
(The trace term $R_0\,\delta^2\Phi\propto(1-3w)\rho\,\bar\zg\zg$ feeds only the
spin-$\tfrac12$ gamma-trace, and vanishes in the radiation era $w=\tfrac13$.)

\paragraph{Non-perturbative activation.}
The transverse spin-$\tfrac32$ component remains an auxiliary field: with no
kinetic term and the mass~\eqref{eq:m32}, its equation of motion is algebraic,
$m^2_{3/2}\,\delta\Psi^{(3/2)}=J$, where $J$ collects the cubic sources
(graviton--graviton, condensate-gradient, and stochastic GW couplings of the
parent mechanism~\cite{MaleknejadKopp2026}). Solving and back-substituting is
the elimination that enforces non-dynamics---\emph{provided} $m^2_{3/2}\neq0$.
The response $\delta\Psi^{(3/2)}=J/m^2_{3/2}$ diverges where the mass passes
through zero, i.e.\ at
\begin{equation}
  \boxed{\,p\to0\quad(w\to0,\ \text{matter domination}).\,}
  \label{eq:activation}
\end{equation}
There the linearized elimination is singular, the constraint that freezes the
mode degenerates, and a sourced spin-$\tfrac32$ response can grow large---the
QSL counterpart of the catastrophic-production point at which the gravitino
sound speed vanishes~\cite{KolbLongMcDonough2021,KalloshKofmanLindeVanProeyen2000}.

\paragraph{Physical conditions, and why activation is not a loophole.}
Four conditions must coincide for the would-be activation:
(a)~a \emph{non-vacuum, curved} background, $G^{\mu\nu}_0\neq0$---generic in
cosmology, but with the no-go \emph{exact} in the radiation era, where the
trace correction vanishes ($R_0\propto1-3w=0$) while the transverse modes stay
massive ($p=\rho/3\neq0$);
(b)~an auxiliary mass eigenvalue \emph{crossing zero}---the transverse channel
at $w=0$, with longitudinal/mixed channels at $w$-values fixed by combinations
of $\rho,p$ and the torsion mass $\Meff$;
(c)~\emph{non-adiabaticity}, $|\dot m_{3/2}/m_{3/2}^2|\gtrsim1$, so the field
cannot track its vanishing constraint adiabatically (an adiabatic crossing
leaves it slaved and unexcited);
(d)~a regime where the expansion about the composite background is itself
controlled---which fails once $\Meff\sim H$ and the condensate amplitude is
large, exactly near~\eqref{eq:activation}.
The decisive point is that even at the crossing the mode has \emph{no kinetic
term}: the divergence of $\delta\Psi^{(3/2)}=J/m^2_{3/2}$ signals a breakdown of
the linearized (and indeed of the controlled effective) description---strong
coupling---rather than the appearance of a healthy propagating particle.

\paragraph{No spin-$\tfrac32$ abundance: the activation is a redundancy.}
The would-be excitation produces no relic, for a structural reason. The
transverse $\delta\Psi^{(3/2)}$ is blind to $\Phi$
[Eq.~\eqref{eq:dPhimap}], so it enters the action \emph{only} through the
metric perturbation it induces,
$h^{(\chi)}_{\mu\nu}=\bar\Psi^0_{(\mu}\delta\Psi^{(3/2)}_{\nu)}+\mathrm{c.c.}$
A direct evaluation of this map on the helicity-$\pm\tfrac32$ polarization (for
$k\parallel\hat z$) gives a metric perturbation with a nonzero
transverse-traceless part---$h_{xx}-h_{yy}$ and $h_{xy}\neq0$---together with a
remainder confined to the non-propagating components ($h_{0i}$, the
longitudinal $h_{iz}$, with $h_{00}=0$ and spatial-traceless). Since the only
propagating part of the metric is the graviton, the spin-$\tfrac32$ fluctuation
is a \emph{redundant relabeling of the graviton plus gauge/constraint data}: it
carries no independent quanta. There is accordingly nothing to populate---any
``production'' is graviton production, already accounted for in the parent
mechanism~\cite{MaleknejadKopp2026}---and the mass-crossing~\eqref{eq:activation}
is the degeneration of this redundant $\delta\Psi^{(3/2)}\!\to\!(h_{\rm
TT},\,\text{gauge})$ parametrization, not a physical particle becoming light.
The controlled spin-$\tfrac32$ relic is therefore zero,
$\Omega_{3/2}h^2=0$.

\paragraph{Quantitative activation window and abundance margin.}
The qualitative crossing~\eqref{eq:activation} can be made explicit. With
$8\pi G=1$ one has $H^2=\rho/3$ and $p=w\rho=3wH^2$, so the
pressure-induced mass~\eqref{eq:m32} is $m^2_{3/2}=\xi\,p$ with
$\xi=\Phi_0=\langle\bar\psi\psi\rangle=O(1)$ and
\begin{equation}
  \frac{m^2_{3/2}}{H^2}=3\xi w
  =\frac{\xi}{1+a/a_{\rm eq}},\qquad
  w(a)=\frac{1}{3\,(1+a/a_{\rm eq})},
  \label{eq:m32_over_H}
\end{equation}
for a radiation-plus-matter background, whence $m_{3/2}\propto a^{-2}$
($m^2_{3/2}\propto p\propto\rho_{\rm rad}\propto a^{-4}$) and, exactly,
\begin{equation}
  \mathcal A(a)\equiv\Bigl|\frac{\dot m_{3/2}}{m^2_{3/2}}\Bigr|
  =\frac{2H}{m_{3/2}}
  =2\sqrt{\frac{1+a/a_{\rm eq}}{\xi}}\;\ge\;\frac{2}{\sqrt\xi}.
  \label{eq:adiabaticity}
\end{equation}
For $\xi=O(1)$ the non-adiabaticity parameter of condition (c) satisfies
$\mathcal A\gtrsim2>1$ \emph{throughout} the post-radiation history
(Fig.~\ref{fig:activation}): the would-be activation is not a narrow
resonance but the whole epoch, and is correspondingly harmless. Lacking a
kinetic term, $\mathcal A\gtrsim1$ excites no quanta; and even the most
generous reinterpretation---the activated transverse mode as a genuine
massless species---is bounded by the single largest gravitational-production
episode, the end of inflation at $H_*\lesssim3\times10^{-5}M_{\rm Pl}$
($r\lesssim0.03$):
\begin{equation}
  \frac{\rho_{3/2}}{\rho_{\rm tot}}\;\lesssim\;0.1\,
  \Bigl(\frac{H_*}{M_{\rm Pl}}\Bigr)^2\;\lesssim\;3\times10^{-11},
  \qquad\Delta N_{\rm eff}\lesssim10^{-10},
  \label{eq:dNeff}
\end{equation}
unobservably small, redshifting as dark radiation (the curvature gap closes
in the radiation era, $R_0\propto1-3w=0$) rather than dark matter, and by the
redundancy argument double-counted graviton data in any case. The structural
result $\Omega_{3/2}h^2=0$ thus carries an explicit quantitative margin.

\begin{figure}[t]
  \centering
  \includegraphics[width=0.96\linewidth]{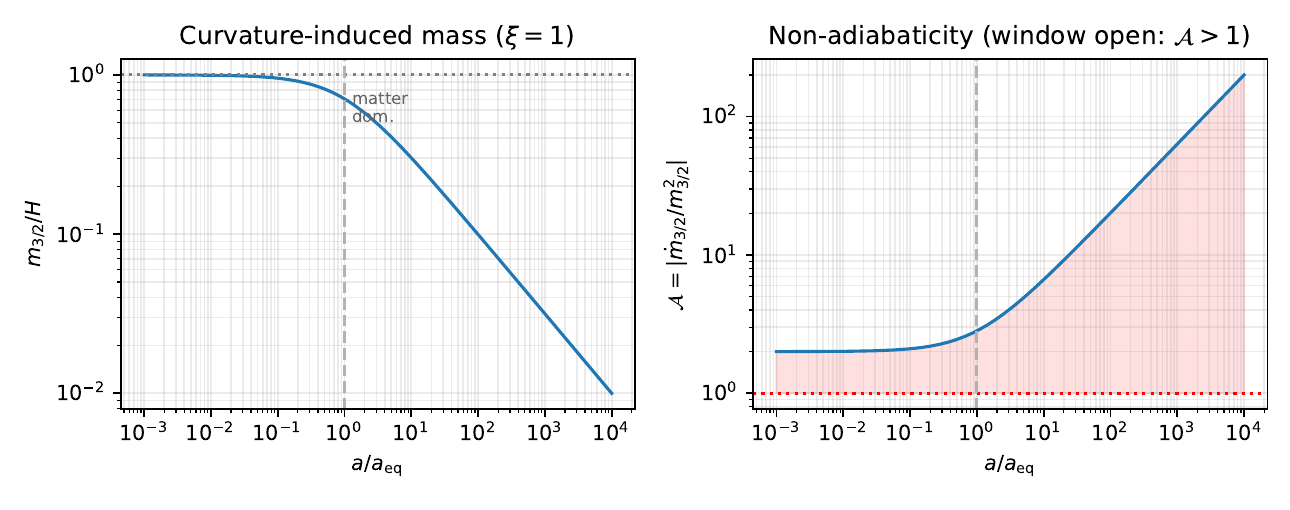}
  \caption{Curvature-induced transverse mass $m_{3/2}/H$ (left) and the
  non-adiabaticity parameter $\mathcal A=|\dot m_{3/2}/m_{3/2}^2|$ (right)
  versus scale factor, for $\xi=1$, on a radiation-plus-matter background
  [Eqs.~\eqref{eq:m32_over_H}--\eqref{eq:adiabaticity}]. The mass falls below
  $H$ as matter domination is approached, and $\mathcal A\ge2$ everywhere:
  the activation window (shaded) is open throughout, yet---there being no
  kinetic term---produces no relic.}
  \label{fig:activation}
\end{figure}

The no-go is thus not evaded but \emph{sharpened}: a spin-$\tfrac32$
excitation can be forced only where the QSL description ceases to be weakly
coupled, and even there it is a redundant description of gravitational data, not
a new dark-matter field. A complete account of the strong-coupling regime
itself---should the background drive the system into it---would require the
cubic action and the parent mechanism's sources, and lies beyond the quadratic
analysis on which the no-go rests; but no controlled, causally propagating
spin-$\tfrac32$ dark matter exists at any point of the cosmological history.

\section{Discussion}
\label{sec:discussion}
The result is a no-go theorem internal to the Quadratic Spinor Lagrangian:
because the action $-\bar\psi\psi\,R$ depends on the field only through the
metric $g=\Psi\otimes_S\Psi$ and the scalar $\bar\psi\psi$, its quadratic
fluctuations factor through a massless metric and a massless scalar, leaving no
room for a propagating massive spin-$\tfrac32$ mode. The conclusion is robust:
it is background-independent (resting on the principal symbol, not on a chosen
frame), independent of the kinetic/mass normalization, and unchanged if the
scalar $\bar\psi\psi$ is itself promoted to an independent field.

The natural comparison is with supergravity, where the gravitino \emph{does}
propagate causally as an Einstein--Cartan field with algebraically sourced
torsion~\cite{DeserZumino1976}. The distinction is precisely the one isolated
here: the gravitino carries an independent Rarita--Schwinger kinetic term,
whereas the QSL's matter sector has none---all its propagation is the
graviton's, channelled through the spinorial frame. The QSL is in this sense
maximally economical: a single Dirac spinor builds the geometry, and the only
excitations are the graviton and that same spinor's spin-$\tfrac12$ quantum.

\subsection{A super-Higgs reading: the spin-$\tfrac12$ as a Goldstino}
\label{sec:superhiggs}
The no-go has a complementary, purely kinematic interpretation through the
super-$SL(2,\mathbb C)$ structure of the QSL, in which the spin connection and the
spinor one-form join in a graded connection with
$\{Q_A,Q_B\}=2M_{AB}$~\cite{Tung2000,TungCQG2026}. The defining bilinear
$g_{\mu\nu}=\bar\Psi_{(\mu}\Psi_{\nu)}$ is an order parameter: a nonzero
$\langle\bar\Psi\Psi\rangle$ transforms as a Lorentz tensor and so \emph{preserves}
the bosonic $SL(2,\mathbb C)$ (the local Lorentz symmetry of the emergent
geometry), while it does \emph{not} annihilate the fermionic generators $Q_A$.
The metric condensate therefore breaks super-$SL(2,\mathbb C)\to SL(2,\mathbb C)$
spontaneously, at the scale $\sqrt F\sim M_{\rm Pl}$ that normalises the
Einstein--Hilbert term.

Under this breaking the Lorentz decomposition of $\Psi_\mu$ (its $\gamma$-trace
spin-$\tfrac12$ part and $\gamma$-traceless spin-$\tfrac32$ part) becomes the
broken-SUSY multiplet. The $\gamma$-trace mode $\gamma^\mu\Psi_\mu$ shifts
inhomogeneously under the broken generators---it is the \emph{Goldstino}---so
Goldstone's theorem forces it to exist and propagate the moment
$\langle\bar\Psi\Psi\rangle\neq0$. The $\gamma$-traceless spin-$\tfrac32$ mode is
the would-be gravitino; lacking any independent Rarita--Schwinger kinetic term (the
very fact established above), it cannot eat the Goldstino in a propagating
super-Higgs, and is instead gapped at $m_{3/2}\sim F/M_{\rm Pl}\sim M_{\rm Pl}$.
The curvature-induced transverse mass $m_{3/2}^2\propto p_{\rm bg}$ of
Sec.~\ref{sec:curved} is then only the residual coupling of those
Planck-gapped components. In this language the present theorem and Goldstone's
theorem are two faces of one statement: the spin-$\tfrac12$ is \emph{forced to
propagate} (Goldstino), while the spin-$\tfrac32$ is \emph{forced not to}
(no kinetic term $\Rightarrow$ Planck-scale gap). Consistently, the Goldstino is
not exactly massless: the slow time-dependence of the cosmological condensate,
$\chi=\chi(t)$, is an explicit breaking that lifts it to a pseudo-Goldstone with
$m=\alpha^{1/2}|\dot\chi/\chi|=|\dot\chi/\chi|/\sqrt6$, recovering precisely the
geometric mass $M_{\rm eff}$ of the companion paper~\cite{TungCQG2026} (the factor
$\sqrt6$ being the same $\alpha^{-1}=6$ of the trace-torsion contorsion). We
present this as an interpretation rather than a second proof, though an explicit
construction supports it: the gravitino supercurrent
$S^\mu=\varepsilon^{\mu\nu\rho\sigma}\gamma_5\gamma_\nu D_\rho\Psi_\sigma$ is
conserved, the longitudinal $\Psi_\sigma\propto q_\sigma$ decouples, and it is the
$\gamma$-trace (spin-$\tfrac12$) mode---not the transverse spin-$\tfrac32$---that
sources $S^\mu$ and its $\gamma$-trace, identifying that mode as the Goldstino with
breaking scale $\sqrt F=\langle\bar\Psi\Psi\rangle^{1/2}=M_{\rm Pl}$ (the
$R$-coefficient of the spinor-curvature identity). Only the $O(1)$ Goldstino
normalisation, fixed by the canonical field map at the condensate, is left to a
complete treatment. Consistently, the divergence $\partial_\mu S^\mu$ on the
trace-torsion background reduces to a time-aligned Dirac mass
$\propto K\gamma^0$ (with $K\propto\dot\chi/\chi$), so the supercurrent's
non-conservation \emph{is} the Goldstino's geometric mass $M_{\rm eff}$ of
Ref.~\cite{TungCQG2026}: a single relation
$\partial_\mu S^\mu\propto F(i\gamma^\nu\partial_\nu-M_{\rm eff})G$ carries both the
breaking scale $F$ and the contorsion factor in $M_{\rm eff}$.

This places the QSL dark-matter candidate on the oldest branch of the
supersymmetric dark-matter tree. That the helicity-$\tfrac12$ components of a
massive gravitino behave, in the high-energy limit, as the Goldstino of global
supersymmetry---with effective derivative coupling $\propto\sqrt{2/3}\,
k^\mu/m_{3/2}$ and $m_{3/2}=\kappa d/\sqrt6=\kappa F/\sqrt3$---is the
gravitino--Goldstino equivalence established by
Fayet~\cite{Fayet1977grav,Fayet1979wk,Fayet1979scat}. We note a striking
structural resonance: the $\sqrt6$ that there relates the gravitino mass to the
supersymmetry-breaking scale is the same $\sqrt6$ that here relates $M_{\rm eff}$
to $|\dot\chi/\chi|$ and normalises the $\gamma$-trace Goldstino, both arising
from the $\gamma$-trace projection of the vector-spinor. An early-decoupling
light gravitino was first identified as a dark-matter candidate by
Fayet~\cite{Fayet1981moriond} and developed by Pagels and
Primack~\cite{PagelsPrimack1982}, stable as the lightest $R$-odd state; for such
a light gravitino it is the helicity-$\tfrac12$ (Goldstino) component that
dominates, so ``Goldstino'' and ``light-gravitino'' dark matter are one object
viewed at different energies. The QSL spin-$\tfrac12$
belongs to this family---it is the Goldstino of a broken local supersymmetry---but
differs from every prior realisation in five respects. (i) The broken symmetry is
\emph{gravitational}: the super-$SL(2,\mathbb C)$ that builds the metric, not a
particle-physics superpartner sector, so there is no sparticle spectrum and no
$R$-parity, and the breaking scale is $M_{\rm Pl}$ rather than a soft mass.
(ii) The Goldstino is \emph{composite}---the spin-$\tfrac12$ mode of the
geometry-building spinor---not an elementary field. (iii) Its mass is the
\emph{geometric}, Hubble-locked $M_{\rm eff}$, a pseudo-Goldstone lift, not a soft
or super-Higgs mass (the would-be $m_{3/2}\sim M_{\rm Pl}$ here merely gaps the
spin-$\tfrac32$ away). (iv) \emph{Only} the Goldstino survives: the transverse
spin-$\tfrac32$ is absent by the theorem above, whereas in standard light-gravitino
dark matter the full multiplet is present and merely Goldstino-dominated.
(v) Production is gravitational~\cite{KolbLongMcDonough2021}, not thermal
regeneration, so the classic gravitino overproduction/BBN tension does not arise.
There is a structural irony here: the helicity-$\tfrac12$/Goldstino component is
precisely the mode whose vanishing sound speed makes massive gravitinos dangerous
to overproduce (the longitudinal enhancement of massive gravitinos discussed
above); the QSL turns that same mode into the dark matter,
while the troublesome transverse modes do not propagate at all.

Several questions remain open. The analysis is linearized about a vacuum
background; whether nonperturbative condensate dynamics could activate
otherwise non-dynamical components is not addressed. The curved-background
extension was argued at the level of the principal symbol and deserves an
explicit Dirac--Bergmann treatment. And the broader question---whether
\emph{any} enlargement of the QSL action could host a healthy spin-$\tfrac32$
field without sacrificing the finite, fully covariant Hamiltonian that
distinguishes the QSL~\cite{NesterTung1995}---lies outside the present scope,
but is sharpened by Remark~\ref{rem:escape}: such an enlargement must add an
operator not built from $(g,\Phi)$, and would thereby define a different
theory.

\section*{Acknowledgments}
\addcontentsline{toc}{section}{Acknowledgments}
I am grateful to Pierre Fayet for drawing my attention to his foundational work
on the gravitino--Goldstino equivalence and on the light gravitino as a
dark-matter candidate, which underlies the discussion of
Section~\ref{sec:superhiggs}, and for his generous and careful reading of this
work.

\appendix
\section{Clifford reduction of the torsional mass}
\label{app:clifford}

We evaluate the derivative-free term
$\mathcal L_{K^2}=2\,\Kc\Psi\,\gamma_5\,\Kc\Psi$ of the torsion-expanded QSL,
first for the composite field (recovering $\alpha=\tfrac16$) and then for the
independent vector-spinor (establishing the frame-aligned
kernel~\eqref{eq:framealigned_mass}). All traces below were verified
symbolically in an explicit Dirac basis by computer
algebra.\footnote{The supporting computer-algebra scripts (\texttt{SymPy})
are available from the author.}

\paragraph{Contorsion and the mass $2$-form.}
For the vectorial trace torsion $T^a=\tfrac13 K\wedge\vartheta^a$ with
$K_a=K\delta^0_a$, the contorsion $1$-form is
$K^{ab}=\tfrac13(K^a\vartheta^b-K^b\vartheta^a)$, and the
Clifford-algebra-valued contorsion is
\begin{equation}
  \Kc=\tfrac14 K_{ab}\gamma^{ab}=\tfrac16 K_a\,\vartheta_b\,\gamma^{ab}
  \;=\;\Kc_e\,\vartheta^e,\qquad
  \Kc_e=\tfrac16\,K_a\,\gamma^{a}{}_{e},\quad
  \gamma^{a}{}_{e}\equiv\gamma^{ab}\eta_{be}.
  \label{eq:Kcomp}
\end{equation}
Writing the spinor $1$-form as $\Psi=\Psi_c\,\vartheta^c$, the product
$\Kc\Psi=\Kc\wedge\Psi$ is the spinor-valued $2$-form $\Xi=\tfrac12\,
\Xi_{ec}\,\vartheta^e\wedge\vartheta^c$ with components
\begin{equation}
  \Xi_{ec}=\Kc_e\,\Psi_c-\Kc_c\,\Psi_e .
  \label{eq:Xi}
\end{equation}
The $\gamma_5$-pairing of two such $2$-forms collapses to the scalar density
\begin{equation}
  \mathcal L_{K^2}
  =\tfrac12\,\epsilon^{pqmc}\,\overline{\Xi_{pq}}\,\gamma_5\,\Xi_{mc},
  \label{eq:LK2}
\end{equation}
the $\epsilon$ arising from expressing the wedge of two $2$-forms as a
volume $4$-form; the Dirac conjugation $\overline{(\,\cdot\,)}$ and the chiral
element $\gamma_5$ are the QSL reality/duality data
[cf.~Eq.~\eqref{eq:identity}].

\paragraph{Composite field.}
Setting $\Psi_c=\gamma_c\psi$ gives $\Xi_{ec}=M_{ec}\psi$ with
\begin{equation}
  M_{ec}=\tfrac16\,K_a\bigl(\gamma^{a}{}_{e}\gamma_c-\gamma^{a}{}_{c}\gamma_e\bigr),
  \label{eq:Mcomp}
\end{equation}
so that $\mathcal L_{K^2}=\bar\psi\,\mathcal S\,\psi$ with
$\mathcal S=\tfrac12\,\epsilon^{pqmc}\,\overline{M_{pq}}\,\gamma_5 M_{mc}$.

To make $M_{ec}$ explicit we use the Clifford split
\begin{equation}
  \gamma^{ab}\gamma^c=\gamma^{abc}+\eta^{bc}\gamma^a-\eta^{ac}\gamma^b,
  \label{eq:split_id}
\end{equation}
which gives the mixed symbol $\gamma^a{}_e=\gamma^{ab}\eta_{be}$ entering the
contorsion~\eqref{eq:Kcomp}. For the cosmological torsion $K_a=K\delta^0_a$
only $a=0$ contributes, and since $\gamma^{00}=0$,
\begin{equation}
  \gamma^0{}_0=0,\qquad \gamma^0{}_i=\gamma^{0i}\ \ (i=1,2,3),
  \label{eq:gmix0}
\end{equation}
so the contorsion is purely the boost-type generator
$\Kc_e=\tfrac16 K\,\gamma^0{}_e$, nonzero only for spatial $e$. The mass
coefficient~\eqref{eq:Mcomp} then reads
\begin{equation}
  M_{ec}=\tfrac16 K\bigl(\gamma^0{}_e\,\gamma_c-\gamma^0{}_c\,\gamma_e\bigr),
\end{equation}
a sum of products $\gamma^{0i}\gamma_c$, each reducible by~\eqref{eq:split_id}
to a totally antisymmetric triple plus a single $\gamma$. Inserting this into
$\mathcal S$ and using $\gamma_5\gamma^{abc}=-i\,\epsilon^{abcd}\gamma_d$ to
pair the triple against the $\epsilon$ tensor, the products collapse onto the
Clifford basis $\{\mathbf 1,\gamma_5,\gamma_a,\gamma_5\gamma_a,\gamma_{ab}\}$,
whose coefficients are isolated by the traces
$\tfrac14\mathrm{tr}\,\mathcal S$, $\tfrac14\mathrm{tr}(\gamma_5\mathcal S)$,
$\tfrac14\mathrm{tr}(\gamma_a\mathcal S)$, and so on. With $K_a=K\delta^0_a$,
only the scalar channel survives:
\begin{equation}
  \mathrm{tr}\,\mathcal S=\tfrac83\,K^aK_a,\qquad
  \mathrm{tr}(\gamma_5\mathcal S)=\mathrm{tr}(\gamma_a\mathcal S)
  =\mathrm{tr}(\gamma_5\gamma_a\mathcal S)=\mathrm{tr}(\gamma_{ab}\mathcal S)=0,
  \label{eq:Straces}
\end{equation}
hence
\begin{equation}
  \mathcal S=\tfrac23\,K^aK_a\,\mathbf 1 .
  \label{eq:Sresult}
\end{equation}
There is no pseudoscalar, vector, axial or tensor part: the trace-torsion mass
is a pure Dirac scalar, with no $\gamma_5$ mixing and no cross terms. Writing
the coefficient against the vector-spinor norm
$\bar\Psi_\mu\Psi^\mu=\bar\psi\gamma_\mu\gamma^\mu\psi=4\bar\psi\psi$ defines
$\alpha$ via $\mathcal S=4\alpha\,K^aK_a\,\mathbf 1$, so
\begin{equation}
  \boxed{\ \alpha=\tfrac16\ },\qquad
  \Meff^2=-\alpha\,K^aK_a=\tfrac16 K^2
  =\tfrac16\Bigl(\tfrac{\dot\varphi}{\varphi}\Bigr)^2,
  \label{eq:alpha}
\end{equation}
using $K^aK_a=-K^2<0$ for the timelike trace vector. (The overall factor $i$
generated by the $\gamma_5$--$\epsilon$ pairing,
$\gamma_5\propto i\,\epsilon_{abcd}\gamma^{abcd}$, multiplies kinetic and mass
densities alike and cancels in the ratio $\Meff^2$; it is the universal QSL
reality convention. The intermediate $\mathcal S=\tfrac23 K^aK_a\mathbf1$,
not $\tfrac16 K^aK_a\mathbf1$, is the correctly normalized scalar coefficient,
with $\Meff^2$ read off through the $\bar\Psi_\mu\Psi^\mu$ bilinear.)

\paragraph{Independent field.}
Retaining all components, $\Psi_c$ is unconstrained and~\eqref{eq:Xi},
\eqref{eq:LK2} give the bilinear
$\mathcal L_{K^2}=\bar\Psi_r\,M^{rs}\,\Psi_s$ with
\begin{equation}
  M^{rs}=\tfrac12\,\epsilon^{pqmc}\Bigl[
   \delta^r_q\,\overline{\Kc_p}\,\gamma_5\,\Kc_m\,\delta^s_c
  -\delta^r_q\,\overline{\Kc_p}\,\gamma_5\,\Kc_c\,\delta^s_m
  -(p\!\leftrightarrow\!q)\Bigr],
  \label{eq:Mrs_raw}
\end{equation}
the four terms coming from expanding both $\overline{\Xi_{pq}}$ and $\Xi_{mc}$
via~\eqref{eq:Xi}. Evaluating with $\Kc_e=\tfrac16 K\gamma^0{}_e$ and
projecting each block $M^{rs}$ on the Clifford basis as above, one finds that
the time--time and all spatial--spatial blocks vanish and only the
time--space blocks are nonzero:
\begin{equation}
  M^{00}=0,\qquad M^{ij}=0,\qquad
  M^{0i}=\kappa\,\gamma^{0i}\quad(\kappa\propto K^2),
  \label{eq:Mrs_blocks}
\end{equation}
i.e.\ in covariant form, with the trace-torsion velocity $u^a=\delta^a_0$,
\begin{equation}
  M^{rs}=\kappa\,u_a\bigl(u^{r}\gamma^{a s}-u^{s}\gamma^{a r}\bigr).
  \label{eq:Mrs_cov}
\end{equation}
This is the result quoted as~\eqref{eq:framealigned_mass}. It is \emph{not}
the uniform scalar $\Meff^2\,\eta^{rs}\mathbf 1$ (which would populate the
diagonal blocks $M^{00},M^{ij}$) and \emph{not} the Rarita--Schwinger form
$\propto\gamma^{rs}$ (which would also carry spatial--spatial blocks); the
torsional mass of the independent field is a frame-aligned operator that
couples only the time component $\Psi_0$ to the spatial components $\Psi_i$.

\paragraph{Consistency.}
Restricting~\eqref{eq:Mrs_cov} to the composite locus $\Psi_c=\gamma_c\psi$,
for which $\bar\Psi_r=\bar\psi\,\overline{\gamma_r}$, the contraction collapses
to the scalar~\eqref{eq:Sresult},
\begin{equation}
  \sum_{r,s}\overline{\gamma_r}\,M^{rs}\,\gamma_s=\tfrac23\,K^aK_a\,\mathbf 1,
\end{equation}
reproducing the composite mass and fixing the normalization $\kappa$. Two
corollaries used in the main text follow at once: a transverse, gamma-traceless
fluctuation ($\Psi_0=0$, $\gamma^i\Psi_i=0$) receives \emph{zero} mass from
$M^{rs}$, since every nonvanishing block carries the index $0$; and the entire
torsional mass is supported on the time component $\Psi_0$, the Lagrange
multiplier of Sec.~\ref{sec:constraints}.

\section{Vanishing of the naive kinetic and cross terms}
\label{app:vanish}

We expand the component form of the QSL $4$-form,
$\mathcal L=\tfrac12\,\epsilon^{pqmc}\,\overline{G_{pq}}\,\gamma_5\,G_{mc}$,
with the curl $G_{ec}=\hat D_e\Psi_c-\hat D_c\Psi_e$ and
$\hat D_e=\partial_e+\Kc_e\to in_e+\Kc_e$ in momentum space, and show that its
purely kinetic ($K^0$) and cross ($K^1$) parts vanish identically, leaving only
the algebraic mass ($K^2$) of Appendix~\ref{app:clifford}.

\paragraph{Decomposition.}
Write $G=iN+\Xi$, separating the momentum curl from the contorsion piece,
\begin{equation}
  N_{ec}=n_e\Psi_c-n_c\Psi_e,\qquad
  \Xi_{ec}=\Kc_e\Psi_c-\Kc_c\Psi_e ,
\end{equation}
both spinor-valued $2$-forms ($N$ real in $n$, $\Xi$ as in
Eq.~\eqref{eq:Xi}). Since $n_e$ is real and the Dirac conjugation reverses the
explicit factor $i$, $\overline{iN_{ec}}=-i\,\overline{N_{ec}}$, so
\begin{equation}
  \mathcal L=\tfrac12\epsilon^{pqmc}\,\overline{(iN+\Xi)}_{pq}\,\gamma_5\,(iN+\Xi)_{mc}
  =\underbrace{\tfrac12\epsilon\,\overline N\gamma_5 N}_{O(K^0)}
  +\underbrace{\tfrac{i}{2}\epsilon\bigl(\overline\Xi\gamma_5 N-\overline N\gamma_5\Xi\bigr)}_{O(K^1)}
  +\underbrace{\tfrac12\epsilon\,\overline\Xi\gamma_5\Xi}_{O(K^2)} .
  \label{eq:Lorders}
\end{equation}

\paragraph{The kinetic term ($K^0$) vanishes.}
Expanding $N_{pq}=n_p\Psi_q-n_q\Psi_p$,
\begin{equation}
  \epsilon^{pqmc}\,\overline{N_{pq}}\,\gamma_5\,N_{mc}
  =4\,\epsilon^{pqmc}\,n_p\,n_m\,\overline{\Psi_q}\,\gamma_5\,\Psi_c ,
\end{equation}
after using the antisymmetry of $\epsilon$ to combine the four cross terms.
Each surviving term carries the two momentum factors $n_p n_m$ on the index
pair $(p,m)$ over which $\epsilon^{pqmc}$ is antisymmetric, while $n_p n_m$ is
symmetric; hence
\begin{equation}
  \epsilon^{pqmc}\,n_p\,n_m=0 \quad\Longrightarrow\quad \mathcal L\big|_{K^0}=0 .
  \label{eq:K0vanish}
\end{equation}
The naive kinetic $4$-form is identically zero---it is the boundary term
$d[\cdots]$ of the spinor--curvature identity~\eqref{eq:identity}, carrying no
bulk dynamics.

\paragraph{The cross term ($K^1$) vanishes.}
The $\gamma_5$ pairing of two spinor-valued $2$-forms is
symmetric~\cite{NesterTung1995},
\begin{equation}
  \epsilon^{pqmc}\,\overline{X_{pq}}\,\gamma_5\,Y_{mc}
  =\epsilon^{pqmc}\,\overline{Y_{pq}}\,\gamma_5\,X_{mc},
  \label{eq:g5sym}
\end{equation}
for any even-form $X,Y$ (it is the same symmetry that makes the QSL action
quadratic and underlies the field equations of Ref.~\cite{TungCQG2026}).
Applied with $X=\Xi$, $Y=N$, identity~\eqref{eq:g5sym} makes the two terms in
the $O(K^1)$ bracket of~\eqref{eq:Lorders} equal, so their difference---and
hence the cross term---vanishes,
\begin{equation}
  \mathcal L\big|_{K^1}
  =\tfrac{i}{2}\,\epsilon^{pqmc}\bigl(\overline{\Xi_{pq}}\gamma_5 N_{mc}
  -\overline{N_{pq}}\gamma_5\Xi_{mc}\bigr)=0 .
  \label{eq:K1vanish}
\end{equation}
The opposite signs in the bracket originate in the single factor $i$ carried by
the momentum curl: the reality convention that renders $2D\Psi\gamma_5 D\Psi$
real is precisely what cancels the term linear in both $n$ and the torsion.

\paragraph{Result.}
Only the derivative-free $O(K^2)$ term survives, reproducing the algebraic
mass kernel of Appendix~\ref{app:clifford}. Thus the naive second-order action
$2D\Psi\gamma_5 D\Psi$ supplies neither a kinetic nor a cross term for the
independent field: its entire content is the non-dynamical torsional mass,
confirming Pillar~II of Sec.~\ref{sec:kinetic_vanish}. Equations
\eqref{eq:K0vanish}--\eqref{eq:K1vanish} were confirmed by direct symbolic
expansion of~\eqref{eq:Lorders} in a Dirac basis.

\section{Second variation, the Lichnerowicz operator, and the mode scan}
\label{app:factor}

This appendix makes the factorized form~\eqref{eq:factorization} explicit and
describes the construction of the inverse propagator whose null-space scan is
Table~\ref{tab:nulls}.

\paragraph{The two maps.}
With the composite background $\Psi^0_\mu=\gamma_\mu\psi_0$,
$\bar\psi_0\psi_0=1$, the linear responses~\eqref{eq:hmap}--\eqref{eq:dPhimap}
are the algebraic, $k$-independent maps
\begin{align}
  h_{\mu\nu}&=\bar\Psi^0_{(\mu}\delta\Psi_{\nu)}+\mathrm{c.c.}
   \ \equiv\ (T\,\delta\Psi)_{\mu\nu},
   &\dim_{\mathbb R}\mathrm{im}\,T&=10,\quad \dim_{\mathbb R}\ker T=22,
   \label{eq:Tmap}\\
  \delta\Phi&=\bar\psi_0\,\tfrac14\gamma^\mu\delta\Psi_\mu+\mathrm{c.c.}
   \ \equiv\ p\cdot\delta\Psi,
   &\mathrm{supp}\,p&=\{\zg=\gamma^\mu\delta\Psi_\mu\}.
   \label{eq:pmap}
\end{align}
Over the real components of $\delta\Psi$ the map $T$ reaches \emph{all} ten
metric components ($\dim_{\mathbb R}\mathrm{im}\,T=10$); its
$22$-real-dimensional kernel---the directions of $\delta\Psi$ that leave the
metric unchanged---is the algebraic origin of the non-dynamical (auxiliary)
sector, and $p$ is supported on the gamma-trace alone. (The inverse propagator
below is assembled in the holomorphic mode basis of Table~\ref{tab:nulls},
where the corresponding complex map has rank $9$; this is a bookkeeping choice
and does not affect the conclusion.)

\paragraph{Linearized Einstein operator.}
The first term of~\eqref{eq:factorization} is, in momentum space, the
linearized Einstein tensor $\mathcal G^{(1)}_{\mu\nu}(h)=\mathcal E(k)\,h$,
\begin{equation}
  \mathcal G^{(1)}_{\mu\nu}
  =-\tfrac12\Bigl[k^2 h_{\mu\nu}-k_\mu k^\alpha h_{\alpha\nu}
   -k_\nu k^\alpha h_{\alpha\mu}+k_\mu k_\nu h
   -\eta_{\mu\nu}\bigl(k^2 h-k^\alpha k^\beta h_{\alpha\beta}\bigr)\Bigr],
  \quad h\equiv h^\alpha{}_\alpha,
  \label{eq:G1}
\end{equation}
which is transverse and gauge invariant: for $h_{\mu\nu}=k_\mu\xi_\nu+k_\nu\xi_\mu$
one finds $\mathcal G^{(1)}_{\mu\nu}=0$ identically, the linearized
diffeomorphism invariance. Its principal part is $k^2$, with no zeroth-order
term.

\paragraph{The scalar--tensor vertex.}
The second term of~\eqref{eq:factorization} is the mixing
$-2\,\delta\Phi\,\delta R$ with the linearized Ricci scalar
\begin{equation}
  \delta R=-k^\mu k^\nu h_{\mu\nu}+k^2 h\ \equiv\ r(k)\cdot\delta\Psi ,
  \label{eq:dR}
\end{equation}
again homogeneous of degree two in $k$. Assembling, the inverse propagator on
the sixteen components of $\delta\Psi$ is the Hermitian form
\begin{equation}
  \mathcal K(k)=T^\dagger\,\mathcal E(k)\,T
  \;-\;\bigl(p^\dagger r(k)+r(k)^\dagger p\bigr)\;+\;\mathcal M,
  \label{eq:Kfull}
\end{equation}
with $\mathcal M$ the algebraic torsional mass of Appendix~\ref{app:clifford}
($k$-independent, supported on $\Psi_0$). Because $\mathcal E(k)$, $r(k)$ are
degree-two and $\mathcal M$ degree-zero in $k$, the only $k$-dependence of
$\mathcal K$ is the homogeneous degree-two kinetic part: its characteristic
variety is contained in $\{k^2=0\}$, the content of Theorem~\ref{thm:nogo}.

\paragraph{Mode scan.}
We computed $\dim\ker\mathcal K(k)$ symbolically over a grid of $k$ and several
backgrounds. Off the cone the null space is the union of the
metric-invisible (auxiliary) directions of $\ker T$ and the gravitational
gauge/constraint directions; on the cone $k^2=0$ an additional massless graviton mode appears,
raising the null dimension by the jump recorded in Table~\ref{tab:nulls}.
Crucially, the null dimension is identical for spacelike and timelike $k^2$ and
jumps \emph{only} at $k^2=0$: there is no timelike (massive) pole. Including
$\mathcal M$ leaves the scan unchanged, since the mass acts only within the
non-propagating $\Psi_0$ sector. The generic null value ($9$ or $10$) depends
on the orientation of $\psi_0$ relative to $\hat k$ and counts the gauge plus
non-dynamical directions; the invariant statement---no jump at timelike
$k^2$---is background- and direction-independent, in agreement with the
analytic Theorem~\ref{thm:factorization}, whose factorization holds at any
$R_0=0$ background.

\paragraph{Remark on rigor.}
The scan is a corroboration, not the proof: the no-go rests on
Theorem~\ref{thm:factorization} (an identity, valid for arbitrary background)
together with the masslessness of $\mathcal Q$. The explicit construction here
fixes a particular $\psi_0$ and momentum direction only to exhibit the mode
content concretely; the table samples three inequivalent backgrounds and three
directions to confirm the orientation-independence of the conclusion. All
null-dimension computations were performed by computer algebra.

\end{document}